\def\confversion{0}
\def\ifconf{\ifnum\confversion=1}
\def\ifnotconf{\ifnum\confversion=0}

\documentclass[11pt]{article}

\def\showauthornotes{0}
\def\showkeys{0}
\def\showdraftbox{0}

\usepackage[utf8]{inputenc} 
\usepackage[T1]{fontenc} 
\usepackage{ae} 
\usepackage{aecompl} 

\usepackage{amsmath} 
\usepackage{amsthm} 
\usepackage{amsfonts} 
\usepackage{amssymb} 
\usepackage{stmaryrd} 
\usepackage{mathtools} 
\usepackage{mathdots} 
\usepackage{bm} 
\usepackage{bbm} 

\usepackage[inline]{enumitem} 
\usepackage{array} 

\usepackage{tikz} 
\usetikzlibrary{cd} 

\usepackage[in]{fullpage} 
\usepackage{setspace} 
\usepackage{indentfirst} 

\usepackage{boththeorems} 

\def\({\left(}
\def\){\right)}
\def\[{\left[}
\def\]{\right]}
\def\<{\left\langle}
\def\>{\right\rangle}

\def\implies{\Longrightarrow}

\newcommand{\place}{{{}\cdot{}}}

\let\geq\geqslant
\let\leq\leqslant

\let\:\colon

\let\epsilon\varepsilon

\def\FF{\ensuremath{\mathbb{F}}}

\def\NN{\ensuremath{\mathbb{N}}}

\def\RR{\ensuremath{\mathbb{R}}}

\def\One{\ensuremath{\mathbbm{1}}}

\def\cC{\ensuremath{\mathcal{C}}}

\usepackage{xspace,xcolor}
\usepackage{amsmath,amssymb}
\usepackage{amsthm}
\usepackage[toc,page]{appendix}
\usepackage{thmtools}
\usepackage{thm-restate}
\usepackage{color,graphicx}
\usepackage{boxedminipage}
\usepackage{makecell}
\usepackage{tabularx}

\ifnum\showkeys=1
\usepackage[color]{showkeys}
\fi

\definecolor{darkred}{rgb}{0.5,0,0}
\definecolor{darkgreen}{rgb}{0,0.35,0}
\definecolor{darkblue}{rgb}{0,0,0.55}

\usepackage[pdfstartview=FitH,pdfpagemode=UseNone,colorlinks,linkcolor=darkblue,filecolor=darkred,citecolor=darkgreen,urlcolor=darkred,pagebackref]{hyperref}

\usepackage[capitalise,nameinlink]{cleveref}
\usepackage[T1]{fontenc}
\usepackage{mathtools,dsfont,bbm}
\usepackage{mathpazo}

\ifnum\showauthornotes=1
\newcommand{\Authornote}[2]{{\sf\small\color{red}{[#1: #2]}}}
\newcommand{\Authorcomment}[2]{{\sf \small\color{gray}{[#1: #2]}}}
\newcommand{\Authorfnote}[2]{\footnote{\color{red}{#1: #2}}}
\else
\newcommand{\Authornote}[2]{}
\newcommand{\Authorcomment}[2]{}
\newcommand{\Authorfnote}[2]{}
\fi

\ifnum\showdraftbox=1
\newcommand{\draftbox}{\begin{center}
  \fbox{%
    \begin{minipage}{2in}%
      \begin{center}%
        \begin{Large}%
          \textsc{Working Draft}%
        \end{Large}\\
        Please do not distribute%
      \end{center}%
    \end{minipage}%
  }%
\end{center}
\vspace{0.2cm}}
\else
\newcommand{\draftbox}{}
\fi

\def\to{\rightarrow}

\def\epsilon{\varepsilon}

\def\phi{\varphi}

\def\implies{\Rightarrow}

\newcommand{\ie}{i.e.,\xspace}

\newcommand{\N}{{\mathbb{N}}}

\newcommand{\F}{{\mathbb F}}

\usepackage{nicefrac}

\newcommand{\abs}[1]{\ensuremath{\left\lvert #1 \right\rvert}}

\newcommand{\one}{{\mathbf{1}}}

\newcommand{\Psymb}{\mathbb{P}}

\DeclareMathOperator*{\pPr}{\widetilde{\Psymb}}

\newcommand{\csos}{\widetilde{\mathcal{C}}}
\newcommand{\al}{\alpha}

\makeatletter

\def\ProbabilityRender#1#2{
  \@ifnextchar\bgroup%
  {\renderwithdist{#1}{#2}}
   {\singlervrender{#1}{#2}}
}
\def\singlervrender#1#2{%
   \ensuremath{\mathchoice
       {{#1}\left[ #2 \right]}
       {{#1}[ #2 ]}
       {{#1}[ #2 ]}
       {{#1}[ #2 ]}
   }
}
\def\renderwithdist#1#2#3{%
   \@ifnextchar\bgroup
   {\superfancyrender{#1}{#2}{#3}}
   {\ensuremath{\mathchoice
      {\underset{#2}{#1}\left[ #3 \right]}
      {{#1}_{#2}[ #3 ]}
      {{#1}_{#2}[ #3 ]}
      {{#1}_{#2}[ #3 ]}
     }
   }
}
\def\superfancyrender#1#2#3#4#5{
   \ensuremath{\mathchoice
      {\underset{#1}{{#1}}\left#4 #3 \right#5}
      {{#1}_{#2}#4 #3 #5}
      {{#1}_{#2}#4 #3 #5}
      {{#1}_{#2}#4 #3 #5}
   }
}
\makeatother

\newfont{\inhead}{eufm10 scaled\magstep1}

\DeclareMathOperator\supp{Supp}

\DeclareMathOperator{\Span}{\operatorname {span}}

\DeclareMathOperator{\im}{\operatorname{im}}

\newcommand{\el}{\ensuremath{\ell} }

\newcommand{\val}{{\sf val}}

\newcommand{\KLP}{\textup{KrawtchoukLP}}
\newcommand{\PKLP}{\textup{PartialKrawtchoukLP}}
\newcommand{\FKLP}{\textup{FullKrawtchoukLP}}

\newcommand{\Cc}{\mathcal{C}}

\newtheoremstyle{plain}
  {\medskipamount}
  {\smallskipamount}
  {\slshape}
  {0pt}
  {\bfseries}
  {.}
  { }
  {\thmname{#1}\thmnumber{ #2}{\normalfont\thmnote{ (#3)}}}

\theoremstyle{plain}

\newboththeorems{theorem}{Theorem}[section]
\newboththeorems{lemma}[theorem]{Lemma}
\newboththeorems{proposition}[theorem]{Proposition}
\newboththeorems{corollary}[theorem]{Corollary}
\newboththeorems{afirmation}[theorem]{Afirmation}
\newboththeorems{fact}[theorem]{Fact}
\newboththeorems{claim}[theorem]{Claim}
\newboththeorems{conjecture}[theorem]{Conjecture}
\newboththeorems{problem}[theorem]{Problem}
\newboththeorems{question}[theorem]{Question}
\newboththeorems{observation}[theorem]{Observation}
\newboththeorems{remark}[theorem]{Remark}
\newboththeorems{definition}[theorem]{Definition}
\newboththeorems{notation}[theorem]{Notation}
\newboththeorems{example}[theorem]{Example}
\newboththeorems{exercise}[theorem]{Exercise}

{\end{proof}\endgroup}

\newenvironment{proofsketch}[1][Proof (sketch)]{%
\begingroup
\begin{proof}[#1]}%
{\end{proof}\endgroup}

{\end{proofsketch}}

{\end{proof}\endgroup}

{\end{proof}\endgroup}

{\end{enumerate}\medskip\endgroup}

\usepackage{thm-restate} 
\usepackage{float}
\usepackage{etoolbox}
\usepackage{empheq}
\AtBeginEnvironment{quote}{\par\singlespacing\small}
\usepackage{xfrac}

\setlist[enumerate]{label={\roman*.}, ref={(\roman*)}}

\let\emph\textit
\def\supp{\operatorname{supp}}

\newcommand{\config}{\textup{\textsf{Config}}}
\newcommand{\forbconfig}{\textup{\textsf{ForbConfig}}}

\newcommand{\Lin}{\textup{Lin}}

\DeclareMathOperator{\GL}{GL}

\definecolor{orange}{HTML}{FF7F00}

\makeatletter
\def\st@r#1#2{\m@th\ooalign{$\hfil#1*\hfil$\cr$#1#2$}}
\def\starprod{\mathop{\mathpalette\st@r\prod}}
\makeatother

\title{Exact Completeness of LP Hierarchies for Linear Codes}

\author{Leonardo Nagami Coregliano\thanks{{\tt IAS}. {\tt lenacore@ias.edu}. This material is based upon work supported by the National Science Foundation, and by the IAS School of Mathematics.} \and
        Fernando Granha Jeronimo\thanks{{\tt IAS}. {\tt granha@ias.edu}. This material is based upon work supported by the National Science Foundation under Grant No. CCF-1900460.} \and
        Chris Jones \thanks{{\tt UChicago}. {\tt csj@uchicago.edu}. This material is based upon work supported by the National Science Foundation under Grant No. CCF-2008920.}}
\date{\today}

\begin{document}
\maketitle
\pagenumbering{roman}

\draftbox

\begin{abstract}
  Determining the maximum size $A_2(n,d)$ of a binary code of
  blocklength $n$ and distance $d$ remains an elusive open question
  even when restricted to the important class of linear codes.
  Recently, two linear programming hierarchies extending Delsarte's LP
  were independently proposed to upper bound $A_2^{\textup{Lin}}(n,d)$
  (the analogue of $A_2(n,d)$ for linear codes).
  One of these hierarchies, by the authors,
  was shown to be \emph{approximately} complete in the sense that the
  hierarchy converges to $A_2^{\textup{Lin}}(n,d)$ as the level grows
  beyond $n^2$.
  Despite some structural similarities, not even approximate
  completeness was known for the other hierarchy by Loyfer and Linial.

  In this work, we prove that both hierarchies recover the
  \emph{exact} value of $A_2^{\textup{Lin}}(n,d)$ at level $n$.
  We also prove that at this level the polytope of Loyfer and Linial
  is integral.
  Even though these hierarchies seem less powerful than general
  hierarchies such as Sum-of-Squares, we show that they have enough
  structure to yield exact completeness via pseudoprobabilities.
\end{abstract}

\thispagestyle{empty}

\newpage

\pagenumbering{arabic}
\setcounter{page}{1}

\section{Introduction}

A binary code is any subset of binary strings $\Cc \subseteq
\F^n_2$. Two fundamental parameters of a code are the size $\abs{\Cc}$
and the minimum (Hamming) distance $d$ between pairs of distinct
codewords. Determining the maximum size $A_2(n,d)$ of a binary code of
blocklength $n$ and distance $d$ remains an elusive open problem
despite much effort and interest in this fundamental
question~\cite{vanLint99,GRS:coding:notes,V19}.

When the distance is $d \coloneqq \lfloor \delta n \rfloor$ for some
constant $\delta \in (0,1/2)$, the growth of $A_2(n,d)$ is known to be
exponential in $n$. It is then convenient to consider the asymptotic
rate $R_2(\delta)$ defined as
\begin{align*}
  R_2(\delta)
  & \coloneqq
  \limsup_{n \to \infty} \frac{1}{n} \log_2 \left( A_2(n,\lfloor \delta n \rfloor)\right).
\end{align*}
Roughly speaking, the maximum size of a code grows as $2^{R_2(\delta)
  n}$ up to lower order terms.  However, the precise asymptotic rate
function $R_2(\delta)$ remains unknown, so this exponential growth is
not fully understood.

The best lower bound on $R_2(\delta)$ dates back to the work of
Gilbert \cite{G52} and Varshamov\footnote{Varshamov showed Gilbert's
  bound for general codes remains the same for linear codes.}
\cite{V57}. Their bound, known as the GV bound, follows from a simple
argument for the distance versus rate trade-off of random codes. The
best upper bound on $R_2(\delta)$ dates back to the work of McEliece,
Rodemich, Rumsey and Welch (MRRW) \cite{MRRW77} and it is based on
linear programming (LP) techniques. Specifically, $A_2(n,d)$ is upper
bounded by the value of an LP of
Delsarte~\cite{delsarte1973algebraic}, which they upper bound by
constructing a dual solution using the theory of orthogonal
polynomials.

Here, we will focus on the family of Delsarte's LPs used in the
so-called first MRRW bound\footnote{In~\cite{MRRW77}, they also
  analyze (in the second MRRRW bound) another family of LPs based on
  the Johnson association scheme.}, which is based on the so-called
Krawtchouk polynomials~\cite{vanLint99} and MacWilliams inequalities
\cite{Macwilliams63,MSG72} since this family is closer to our work.
The precise details of this family of LPs are not important at this
point.

Linear codes (i.e., linear subspaces) are arguably one of the most
important and widely studied classes of
codes~\cite{vanLint99,GRS:coding:notes}. We denote by $A_2^\Lin(n,d)$
and $R_2^\Lin(\delta)$ the versions of $A_2(n,d)$ and $R_2(\delta)$,
respectively, corresponding to linear codes. Even for this important
class of codes, the known lower and upper bounds for $R_2^\Lin(\delta)$
are the same as those for $R_2(\delta)$ for general codes.

Delsarte's linear programs are a convex relaxation for $A_2(n,d)$ and
there is a known gap between the value of the LP and the GV
bound~\cite{Samorodnitsky01,NS05}.  In other words, if the GV
bound is indeed tight, then Delsarte's LP is not sufficient to prove
it (this would be called an integrality gap of the LP).  For this reason, it is natural to look for approaches that are
provably sufficient to settle the growth of $A_2(n,d)$ while having
the hope of being amenable to theoretical analysis. Note that
Delsarte's LPs do not distinguish between general and linear codes,
hence they do not provide better bounds for $A_2^\Lin(n,d)$ nor to the
asymptotic rate $R_2^\Lin(\delta)$.

There have been attempts to improve the upper bound using stronger
convex relaxations of $A_2(n,d)$.  The problem of computing $A_2(n,d)$
is equivalent to computing the independence number of a graph whose
vertex set is $\F_2^n$ and pairs of vertices are adjacent if they
violate the minimum distance constraint. In principle, one can employ
general convex programming hierarchies such as
Sum-of-Squares~\cite{L07} or Sherali--Adams, which provably equal the
true value $A_2(n,d)$ at a sufficiently large level. Delsarte's LP is
equivalent to a convex relaxation for independent set known as
Schrijver's $\vartheta'$ function~\cite{Schrijver79,L07}.  This is a
slight strengthening of the Lov\'{a}sz $\vartheta$
function~\cite{Lov79}, which is equivalent to the first level of the
Sum-of-Squares hierarchy for independent set.  However, analyzing
these general hierarchies remains elusive; in fact it even remains
open to analyze an SDP proposed by Schrijver~\cite{Schrijver05}, which
lies between Delsarte's LP and the second level of the Sum-of-Squares
of hierarchy. The only convex programs we know how to analyze for this
problem are Delsarte's LPs, and there are now
a few different techniques for this
analysis~\cite{MRRW77, FT05,NS05,NavonS09,Samorodnitsky2021proof}.

Recently, two new convex programming hierarchies for $A^\Lin_2(n,d)$
were proposed, one by Coregliano, Jeronimo, and Jones~\cite{CJJ22}
(see also~\cite{J22} for an alternative exposition) and another by
Loyfer and Linial~\cite{LL22} (in fact both hierarchies can be defined
over general finite fields).  In this paper, we study these
hierarchies further.

The two hierarchies are similar in spirit but not exactly the same.
Both hierarchies are a family of LPs (rather than SDPs) that extend
Delsarte's LP into a hierarchy of tighter and tighter convex
relaxations for $A^\Lin_2(n,d)$, while retaining some structural
similarities with Delsarte's LP. Since Delsarte's LP is the only
convex program with known theoretical analysis, there is a hope that
analyzing these two new hierarchies may be possible.

The Krawtchouk hierarchy $\KLP^{\FF_q}_\Lin(n, d, \el)$
of~\cite{CJJ22} was shown to be \emph{approximately} complete beyond
level $n^2$ in the following sense.
The level of the hierarchy is $\el$, where $\el = 1$ recovers
Delsarte's LP.

\begin{theorem}[\cite{CJJ22}]\label{theo:completeness_informal}
  For $\el \ge \Omega_{\epsilon,q}(n^2)$, we have
  \begin{align*}
    A^\Lin_q(n,d) ~\leq~ 
    \val(\KLP^{\FF_q}_\Lin(n, d, \el))^{1/\ell} ~\leq~ (1+\epsilon) \cdot A^\Lin_q(n,d).
  \end{align*}
\end{theorem}

Our first result is the \emph{exact} completeness at level $n$ of the
Krawtchouk hierarchy as follows.

\begin{theorem}\label{thm:completeness_kraw}
  For $\ell \ge n$, we have $A_q^\Lin(n,d) = \val(\KLP_\Lin^{\F_q}(n,d,\ell))^{1/\ell}$.
\end{theorem}

Instead of relying on integrality of the feasible region (\ie it is
exactly the convex hull of true solutions corresponding to linear
codes) to deduce completeness as is the case for general hierarchies
such as Sherali--Adams or Sum-of-Squares, it is only possible to show
that optimum solutions are integral, giving an unusual proof of
completeness for a convex programming hierarchy. We also show that the
polytope of $\KLP_\Lin^{\F_q}(n,d,\ell)$ is \emph{never} integral
(see~\cref{prop:notintegral}). Nonetheless, any given non-integral
solution becomes infeasible as the level grows
(see~\cref{prop:intersectionintegral}).

The partial Krawtchouk hierarchy $\PKLP_\Lin^{\F_q}(n,d,\ell)$
of~\cite{LL22} is similar to the Krawtchouk hierarchy of~\cite{CJJ22},
but it has additional constraints and a different objective function
(see~\cref{sec:partial_kraw_comple}). Due to this different objective
function and the fact that the \emph{approximate} completeness proof
of~\cref{theo:completeness_informal} crucially relies on the objective
function of the Krawtchouk hierarchy being ``dense'', the proof
of~\cref{theo:completeness_informal} did not extend to the Loyfer and
Linial hierarchy. Our proof here of \cref{thm:completeness_kraw} does
extend to show that $\PKLP_\Lin^{\F_q}(n,d,\ell)$ is complete at level
$n$. More precisely, our second result is the following.
  
\begin{theorem}\label{thm:completeness_partial}
  For $\ell \ge n$, we have $A_q^\Lin(n,d) = \val(\PKLP_\Lin^{\F_q}(n,d,\ell))$.
\end{theorem}

Curiously, we show that the additional constraints of the Loyfer and
Linial hierarchy make the polytope integral for $\ell \ge n$
(see~\cref{prop:integ_ll_hier}). This integrality
is not obvious from the original formulation of the hierarchy and it
relies on a new perspective uncovered by this work.

The exact completeness theorems at level $n$
(\cref{thm:completeness_kraw,thm:completeness_partial})
improve our understanding of these hierarchies, consolidating them as
provable approaches to resolve the longstanding question of
improving bounds for $A_2^{\Lin}(n,d)$,
and justifying them as natural objects in their own right.
The primary open research direction is a theoretical
analysis of these hierarchies to obtain tighter bounds on
$R_2^{\Lin}(\delta)$.
It is not clear which hierarchy is better suited for such a task:
the Krawtchouk hierarchy may be simpler to analyze, which is of
critical importance here, but the partial Krawtchouk hierarchy may provide
tighter values at the same level given its additional constraints.

\paragraph{Proof Outline.} We first briefly recall the approximate completeness proof
from~\cite{CJJ22}. The hierarchy $\KLP^{\F_q}_\Lin(n,d,\el)$ can be
seen as a symmetrization of the $\vartheta'$ of a graph from a
carefully chosen association scheme under the actions of symbol
permutation by $S_n$ and translation by $\F_q^n$
(see~\cite{delsarte1973algebraic,DL98} and~\cite[\S 5]{CJJ22} for more
on association scheme theory). The \emph{approximate} completeness is
then obtained via a counting argument over the \emph{unsymmetrized}
$\vartheta'$ formulation, which requires level $\el\geq n^2$ to yield
non-trivial bounds.

A key insight of this work is a novel third formulation of the
Krawtchouk hierarchy from which \emph{exact} completeness can be
obtained at level $n$.
Instead of factoring symmetries that lead to variables indexed by
Hamming weights, we now factor different symmetries leading to
variables indexed by linear subspaces.
Using a linear transformation (namely, M\"{o}bius inversion of the
poset of subspaces of $\F_q^n$), we then rewrite the LP in terms of
new variables that can interpreted as a pseudoprobability distribution
over linear codes (see~\cref{sec:kraw_comple}).
In this pseudoprobability formulation, integral solutions correspond
to true probability distributions and via a mass transfer argument we
show that optimum solutions are integral.
An interesting feature of this third formulation of the hierarchy is that the number of variables
and constraints remains constant regardless of the level (see~\cref{subsec:completeness_kt}).
Curiously, we show (\cref{prop:notintegral}) that the polytope of this
formulation is not integral, \ie there are non-optimum solutions that
are not integral.
As mentioned, these ideas also generalize to show that the partial
Krawtchouk hierarchy of~\cite{LL22} also has \emph{exact}
completeness for $\ell \ge n$ (see~\cref{sec:partial_kraw_comple}).

\noindent \textbf{Bibliographic Note:} A preliminary version of the
completeness of the $\KLP_\Lin^{\F_2}(n,d,\ell)$ hierarchy (over the
binary field) is included in the dissertation of one of the
authors~\cite{J22}.


\section{Preliminaries}

We denote by $\F_q$ the finite field of size $q$ (which must be a
prime power). A \emph{code} of blocklength $n \in \mathbb{N}^+$ (over
$\F_q$) is a non-empty subset $\Cc \subseteq \F_q^n$. We denote by
$\Delta(x,y)\coloneqq\lvert\{i\in[n] \mid x_i\neq y_i\}\rvert$ the
\emph{Hamming distance} between $x,y \in \F_q^n$. The \emph{minimum
  distance} of a code $\Cc$ is the minimum of $\Delta(x,y)$ over all
distinct $x,y\in\Cc$. The \emph{rate} $r(\Cc)$ of $\Cc$ is defined
as $r(\Cc) \coloneqq \log_q(\abs{\Cc})/n$. We denote by $A_q(n,d)$ the
maximum size of a code of blocklength $n$ (over $\F_q$) and minimum
distance at least $d$. We say that $\Cc$ is \emph{linear} if it is an
$\F_q$-linear subspace. For linear codes, we have $r(\Cc) =
\dim_{\F_q}(\Cc)$. We denote by $A_q^\Lin(n,d)$ the analogue of
$A_q(n,d)$ when codes are required to be linear.

For $\al \in \F_q^n$, we denote by $\chi_\al \colon \F_q^n \to
\mathbb{C}$ the (additive) \emph{Fourier character} associated with
$\al$ and we denote by $\one_\al \colon \F_q^n \to \{0,1\}$ be the
\emph{indicator function} of $\al$. If $T$ is a vector space (over
$\F_q$), we will use the notation $S\leq T$ to mean that $S$ is a
subspace of $T$ and $S < T$ to mean that $S$ is a proper subspace of
$T$.

\medskip

The rest of this section is devoted to informal descriptions of the
hierarchies from~\cite{CJJ22} and~\cite{LL22} in their symmetrized
form. Since all arguments of this paper start from the unsymmetrized
versions in~\cref{lp:unsym_klp,program:partial_kt}
of~\cref{sec:kraw_comple,sec:partial_kraw_comple} to factor different
symmetries, the descriptions below serve only as guiding intuition and
will not be used in any proof.

The hierarchy from~\cite{CJJ22} extends Delsarte's LPs by considering
not only the Hamming weight of single codewords, but by also
considering the Hamming weights of every codeword in subspaces of
dimension up to a parameter $\ell \in \N^+$, which is the level of the
hierarchy. Given an $\ell$-tuple of words $(x_1,\ldots,x_\ell) \in
(\F_q^n)^\ell$, one associates a \emph{configuration} function mapping
$c\in\F_q^\el$ to the Hamming weight $\abs{\sum_{j\in[\el]} c_j\cdot
  x_j}$ (in the binary case, it is typical to naturally identify
$\F_2^\el$ with the set $2^{[\el]}$ of subsets of $[\el]$). The set of
functions $\F_q^\el\to\NN$ that are configurations of some $\el$-tuple
of codewords is denoted $\config$. If the words $x_1,\ldots,x_\ell$
belong to some linear code $\Cc$ of minimum distance $d$, then their
configuration cannot have numbers from
$[d-1]\coloneqq\{1,\ldots,d-1\}$ in its image. This means that if we
let $a_g$ be the number of tuples $(x_1,\ldots,x_\el)$ in $\Cc$ whose
configuration is $g$, then $a_g = 0$ whenever
$g\in\forbconfig\coloneqq\{h\in\config \mid
[d-1]\cap\im(h)\neq\varnothing\}$. It is clear that $a_0=1$ for the
zero configuration (as $0\in\Cc$ since $\Cc$ is linear) and that
$\lvert\Cc\rvert^\ell=\sum_{g\in\config} a_g$. Finally, by observing that
the Fourier transform of the indicator $\One_{\Cc}$ is (up to a
multiplicative constant) the indicator of the dual code $\Cc^\perp$,
hence a nonnegative function, one derives the so-called (higher-order)
MacWilliams inequalities based on a higher-order version of the
Krawtchouk polynomials. The level $\ell$ of this hierarchy for codes
over the field $\F_q$ is denoted by $\KLP^{\F_q}_\Lin(n,d,\ell)$ and
it is a relaxation (i.e., an upper bound) for $A_q(n,d)^\ell$. The
program in~\cref{lp:klp} provides an informal description of this
hierarchy, where variables are indexed by configurations and $K_h$ is
the higher-order Krawtchouk polynomial associated with configuration
$h$. In this formulation, it is immediate that the first level of this
hierarchy is simply Delsarte's LP. Since we will work with a different
formulation of the hierarchy (see~\cref{lp:unsym_klp}
in~\cref{sec:kraw_comple}), we point the interested reader
to~\cite{CJJ22,J22} for a more detailed description of this Hamming
weight formulation of the hierarchy.

\begin{figure}[!htb]
{\small
\begin{empheq}[box=\fbox]{align*}    
    \max \quad
    & \sum_{g \in \config} a_g
    \\
    \text{s.t.} \quad
    & a_0 = 1
    & &
    & & (\text{Normalization})
    \\
    & a_g = 0
    & & \forall g \in\forbconfig
    & &  (\text{Distance constraints})
    \\
    & \sum_{g \in \config} K_h(g)\cdot a_g \geq 0
    & & \forall h \in \config
    & & (\text{MacWilliams inequalities})
    \\
    & a_g \geq 0
    & & \forall g \in  \config
    & & (\text{Nonnegativity}).
  \end{empheq}}
  \caption{Informal description of $\KLP^{\F_q}_\Lin(n,d,\ell)$. Its optimum
    value is an upper bound for $A_q^{\Lin}(n,d)^\el$.}\label{lp:klp}
\end{figure}

As we mentioned in the introduction, Loyfer and Linial in~\cite{LL22}
independently proposed another linear programming hierarchy
$\PKLP^{\F_q}_\Lin(n,d,\el)$ for linear codes that bears many structural
similarities with the Krawtchouk hierarchy $\KLP^{\F_q}_\Lin(n,d,\ell)$
of~\cite{CJJ22}, but it is different in two important aspects. Firstly,
$\PKLP^{\F_q}_\Lin(n,d,\ell)$ uses a different objective function that sums only
over configurations in $\config_1\coloneqq\{g\in\config \mid \forall
c\in\supp(g), \{1\}\subseteq\supp(c)\}$ (configurations in $\config_1$
correspond to $\ell$-tuples of codewords of the form $(x_1,0,\ldots,0)$); this
provides an upper bound for $A_q^\Lin(n,d)$ (as opposed to its $\ell$th
power). Secondly, by using partial Fourier transforms as well as the usual
Fourier transform (see~\cite{LL22} or~\cref{sec:partial_kraw_comple} below for
more details), $\PKLP^{\F_q}_\Lin(n,d,\el)$ also has ``partial MacWilliams
inequality'' constraints that are not present in
$\KLP^{\F_q}_\Lin(n,d,\el)$. The program in~\cref{lp:pklp} provides an informal
description of this hierarchy, where variables are indexed by configurations and
$K_h^S$ is the partial higher-order Krawtchouk polynomial associated with
configuration $h$ and set $S\subseteq[\el]$.

\begin{figure}[!htb]
{\small
\begin{empheq}[box=\fbox]{align*}    
    \max \quad
    & \sum_{g \in \config_1} a_g
    \\
    \text{s.t.} \quad
    & a_0 = 1
    & &
    & & (\text{Normalization})
    \\
    & a_g = 0
    & & \forall g \in\forbconfig
    & &  (\text{Distance constraints})
    \\
    & \sum_{g \in \config} K_h^S(g)\cdot a_g \geq 0
    & & \forall h \in \config, \forall S\subseteq[\el]
    & & (\text{Partial MacWilliams inequalities})
    \\
    & a_g \geq 0
    & & \forall g \in  \config
    & & (\text{Nonnegativity}).
  \end{empheq}}
  \caption{Informal description of $\PKLP^{\F_q}_\Lin(n,d,\ell)$. Its optimum
    value is an upper bound for $A_q^{\Lin}(n,d)$. The nonnegativity constraints
    $a_g\geq 0$ are redundant as they are also obtained as the partial
    MacWilliams inequalities corresponding to $S=\varnothing$. The original
    formulation also enforces $\GL_\el(\F_q)$ symmetries, but these are omitted
    here for simplicity.}\label{lp:pklp}
\end{figure}


\section{Exact Completeness of the Krawtchouk LP Hierarchy}\label{sec:kraw_comple}

In this section, we prove the exact completeness at level $n$ of the Krawtchouk hierarchy for linear
codes, namely, we show that $A^\Lin_q(n,d) = \val(\KLP_\Lin^{\F_q}(n,d,n))^{1/n}$. We first give an
alternative formulation of this hierarchy in terms of pseudoprobabilities
in~\cref{subsec:pseudoprob_formulation}. Using this representation, we then show the exact
completeness result in~\cref{subsec:completeness_kt}.

\subsection{A Pseudoprobability LP Formulation}\label{subsec:pseudoprob_formulation}

We first recall the unsymmetrized formulation of the hierarchy from
\cite{CJJ22} given in~\cref{lp:unsym_klp}; it corresponds to the
$\vartheta'$ formulation of the Krawtchouk hierarchy expressed in
``diagonalized'' form using the Fourier basis. Here, we use this
unsymmetrized formulation as our starting point. The interested reader
is referred to~\cite{CJJ22} for more details about the connection
between these equivalent formulations of the hierarchy.

\begin{figure}[!htb]
{\small
\begin{empheq}[box=\fbox]{align*}
  & \text{Variables: } a_x && x \in (\F_q^n)^\el \\
  \max \quad
  & \sum_{x \in (\F_q^n)^\el} a_{x}
  \\
  \text{s.t.} \quad
  & a_{0} = 1
  & &
  & & (\text{Normalization})
  \\
  & a_{(x_1, \dots, x_\el)} = 0
  & & \exists w \in \Span(x_1, \dots, x_\el).\; \abs{w} \in [d-1]
  & & (\text{Distance constraints})
  \\
  & \sum_{x \in (\F_q^n)^\el} a_x\chi_\alpha(x) \ge 0
  & & \forall \alpha \in (\F_q^n)^\el
  & & (\text{Fourier coefficients})
  \\
  & a_x = a_{-x}
  & & \forall x \in (\F_q^n)^\el
  & & (\text{Reflection})
  \\
  & a_x \geq 0
  & & \forall x \in (\F_q^n)^\el
  & & (\text{Nonnegativity}).
\end{empheq}}
  \caption{Unsymmetrized higher-order Krawtchouk hierarchy $\KLP^{\F_q}_\Lin(n,d,\ell)$ for $A_q(n,d)$.}\label{lp:unsym_klp}
\end{figure}

To each linear code $\Cc\le\F_q^n$, we have a corresponding \emph{true solution} $a^\Cc$ given by
\begin{align*}
  a^\Cc_x & \coloneqq \One[\forall j\in[\el], x_j\in\Cc],
\end{align*}
whose value is $\lvert\Cc\rvert^\el$. Note that $a^\Cc$ is feasible for the program
in~\cref{lp:unsym_klp} if and only if $\Cc$ has minimum distance at least $d$.

On the other hand, the program in~\cref{lp:unsym_klp} is invariant under the natural basis change
action of the general linear group $\GL_\el(\F_q)$; this means that by symmetrizing a solution $a$
under such action, we may assume that $a_x = a_y$ whenever $\Span(x) = \Span(y)$; after such
symmetrization, we can denote by $a_S$ ($S\le \F_q^n$) the value of $a_x$ for any $x\in(\F_q^n)^\el$
such that $\Span(x) = S$. Note that the true solutions $a^\Cc$ corresponding to \emph{linear} codes
$\Cc\le\F_q^n$ are already symmetrized:
\begin{align}\label{eq:aCcalt}
  a^\Cc_x & = \one[\Span(x)\subseteq\Cc] \eqqcolon a^\Cc_{\Span(x)}.
\end{align}

Equation~\eqref{eq:aCcalt} above suggests that we should interpret the variables $a_S$ as the
relaxation of the indicator $\one_{S \subseteq \Cc}$ for a code $\Cc$; or more precisely as $a_S =
\pPr[S \subseteq \csos]$, where $\csos$ is a formal variable that represents a code drawn from a
pseudodistribution of linear codes.

The next lemma uses M\"{o}bius inversion to provide a linear transformation into variables of the
form $\pPr[S = \csos]$ and shows that (symmetrized) integral solutions are precisely those in which
$\pPr[S = \csos]$ ($S\le\F_q^n$) is a (true) probability distribution (recall that a solution $a$ is
\emph{integral} if it is a convex combination of true solutions $a^\Cc$).

\begin{lemma}\label{lem:mobius}
  For every $S\le\F_q^n$, let $\pPr[S\subseteq\csos],\pPr[S = \csos]\in\RR$ be real numbers. Then
  the following are equivalent.
  \begin{enumerate}
  \item For every $S\le\F_q^n$, we have $\pPr[S\subseteq\csos] = \sum_{S\le
    T\le\F_q^n}\pPr[T=\csos]$.
    \label{lem:mobius:direct}
  \item For every $S\le\F_q^n$, we have $\pPr[S=\csos] = \sum_{S\le T\le\F_q^n}
    \mu(S,T)\pPr[T\subseteq\csos]$, where $\mu$ is the M\"{o}bius function of the poset of subspaces
    of $\F_q^n$ under inclusion.
    \label{lem:mobius:inverse}
  \end{enumerate}

  Furthermore, if $\pPr[S\subseteq\csos],\pPr[S = \csos]$ ($S\le\F_q^n$) satisfy the above, then for
  every $S\le\F_q^n$, we have
  \begin{align*}
    \pPr[S\subseteq\csos] & = \sum_{T\le\F_q^n} \pPr[T=\csos]\cdot a^T_S.
  \end{align*}
  In particular, the solution $a_S \coloneqq \pPr[S\subseteq\csos]$ ($S\le\F_q^n$) is integral if
  and only if $\pPr[S=\csos]$ ($S\le\F_q^n$) is a probability distribution.
\end{lemma}

\begin{proof}
  Recall that the M\"{o}bius function $\mu$ is inductively defined\footnote{In fact, one can show
    that $\mu(S,T) = (-1)^{\dim(T/S)} q^{\binom{\dim(T/S)}{2}}$ when $S\le T$ (and $0$ when
    $S\not\le T$), but we will not need this explicit formula.} by
  \begin{align*}
    \mu(S,T) & \coloneqq
    \begin{dcases*}
      1, & if $S = T$,\\
      - \sum_{S\leq U < T} \mu(S,U), & if $S < T$,\\
      0, & if $S \not\le T$,
    \end{dcases*}
  \end{align*}
  which in particular means that we have $\sum_{S\le U\le T} \mu(S,U) = \sum_{S\le U\le T} \mu(U,T)
  = \one[S = T]$ for every $S\le T\le\F_q^n$.

  For the implication~\ref{lem:mobius:direct}$\implies$\ref{lem:mobius:inverse}, note that for every
  $S\le\F_q^n$, we have
  \begin{align*}
    \sum_{S\le T\le\F_q^n} \mu(S,T)\pPr[T\subseteq\csos]
    & =
    \sum_{S\le T\le\F_q^n} \mu(S,T)\sum_{T\le U\le\F_q^n}\pPr[U=\csos]
    \\
    & =
    \sum_{S\le U\le\F_q^n}\pPr[U=\csos]\sum_{S\le T\le U}\mu(S,T)
    =
    \pPr[S=\csos].
  \end{align*}

  For the implication~\ref{lem:mobius:inverse}$\implies$\ref{lem:mobius:direct}, note that for every
  $S\le\F_q^n$, we have
  \begin{align*}
    \sum_{S\le T\le\F_q^n} \pPr[T = \csos]
    & =
    \sum_{S\le T\le\F_q^n} \sum_{T\le U\le\F_q^n} \mu(T,U)\pPr[U\subseteq\csos]
    \\
    & =
    \sum_{S\le U\le\F_q^n}\pPr[U\subseteq\csos]\sum_{S\le T\le U}\mu(T,U)
    =
    \pPr[S\subseteq\csos].
  \end{align*}

  For the second assertion, since $a^\Cc_S = \one[S\subseteq\Cc]$, from~\ref{lem:mobius:direct}, we
  have
  \begin{align*}
    \pPr[S\subseteq\csos]
    & =
    \sum_{T\le\F_q^n}\pPr[T=\csos]\cdot\one[S\le T]
    =
    \sum_{T\le\F_q^n}\pPr[T=\csos]\cdot a^T_S,
  \end{align*}
  that is, the solution $\pPr[\place\subseteq\csos]$ is written as the linear combination
  \begin{align*}
    \pPr[\place\subseteq\csos] & = \sum_{T\le\F_q^n}\pPr[T=\csos]\cdot a^T
  \end{align*}
  of the true solutions $a^T$; this linear combination is a convex combination precisely when
  $\pPr[T=\csos]\ge 0$ for every $T\le\F_q^n$ and $\sum_{T\le\F_q^n}\pPr[T=\csos] = 1$.
\end{proof}

The idea of the proof of completeness is to rewrite the linear program in terms of the variables
$\pPr[S = \csos]$ and then argue about the program from the perspective of the pseudoprobabilities.
For simplicity, let us now shorten the notation to $\pPr[S] \coloneqq \pPr[S = \csos]$.

\begin{figure}[!htb]
{\small
\begin{empheq}[box=\fbox]{align*}
  & \text{Variables: } \pPr[S] && S \leq \F_q^n \\
  \max \quad
  & \sum_{S \le \F_q^n} \abs{S}^\el \pPr[S]
  \\
  \text{s.t.} \quad
  & \sum_{S \leq \F_q^n} \pPr[S] = 1
  & &
  & & (\text{Normalization})
  \\
  & \pPr[S] = 0
  & & \exists w \in S.\; \abs{w} \in [d-1] 
  & & (\text{Distance Constraints})
  \\
  & \sum_{S \leq U} \abs{S}^\el\pPr[S] \geq 0
  & & \forall U \leq \F_q^n
  & & (\text{Fourier coefficients})
  \\
  & \sum_{S \geq U}\pPr[S] \geq 0
  & & \forall U \leq \F_q^n
  & & (\text{Nonnegativity}).
\end{empheq}}
 \caption{$\KLP^{\F_q}_\Lin(n,d,\el)$ in terms of pseudoprobabilities for $\el \geq n$.}
 \label{fig:pseudoprobabilities_orig}
\end{figure}

\begin{lemma}\label{lem:pseudoprob}
  If $a$ is a $\GL_\el(\F_q)$-invariant solution of the program $\KLP^{\F_q}_\Lin(n,d,\el)$
  in~\cref{lp:unsym_klp} and $\pPr[\Span(x)\subseteq\csos] \coloneqq a_x$ for every
  $x\in(\F_q^n)^\el$, then $\pPr[S = \csos]$ given by~\cref{lem:mobius}\ref{lem:mobius:inverse} is a
  solution of the program in~\cref{fig:pseudoprobabilities_orig} with the same value.

  Conversely, if $\pPr[S=\csos]$ is a solution of the program
  in~\cref{fig:pseudoprobabilities_orig}, then setting $a_x \coloneqq \pPr[\Span(x)\subseteq\csos]$
  via~\cref{lem:mobius}\ref{lem:mobius:direct} gives a solution of $\KLP^{\F_q}_\Lin(n,d,\el)$ with
  the same value.
\end{lemma}

\begin{proof}
  We rewrite the ($\GL_\el(\F_q)$-symmetrization of) the program $\KLP^{\F_q}_\Lin(n,d,\el)$
  in~\cref{lp:unsym_klp} in terms of the variables $\pPr[S]$ obtained from
  $\pPr[S\subseteq\csos]\coloneqq a_S$ via~\cref{lem:mobius}.

  The rewritten objective function is
  \begin{align*}
    \sum_{x \in (\F_q^n)^\el} \pPr[\Span(x) \subseteq \csos]
    & =
    \sum_{x \in (\F_q^n)^\el} \sum_{\Span(x)\le T\le\F_q^n}\pPr[T]
    =
    \sum_{S \leq \F_q^n} \abs{S}^\el \pPr[S].
  \end{align*}

  The left-hand side of the distance constraint for $S\le\F_q^n$ such that there exists $w \in S$
  with $\abs{w} \in [d-1]$ is
  \begin{align*}
    \pPr[S\subseteq\csos]
    =
    \sum_{T \geq S} \pPr[T]
  \end{align*}
  By induction downwards on the dimension of $S$, requiring the above to be equal to $0$ is
  equivalent to the constraints
  \begin{align*}
    \pPr[S] = 0 \qquad (S \leq \F_q^n: \; \exists w \in S.\; \abs{w} \in [d-1]).
  \end{align*}

  The left-hand side of the Fourier constraint for $\al$ is
  \begin{align*}
    \sum_{x \in (\F_q^n)^\el} \pPr[\Span(x) \subseteq \csos] \chi_\al(x)
    & =
    \sum_{x \in (\F_q^n)^\el} \chi_\al(x) \sum_{\Span(x)\le T\le\F_q^n}\pPr[T]
    =
    \sum_{S \leq \F_q^n} \pPr[S] \sum_{\substack{x \in S^\el}} \chi_\al(x).
  \end{align*}

  Let us now show the following claim.
  \begin{claim}\label{clm:fourier_subspace_sum}
    For $\al \in (\F_q^n)^\el, S \leq \F_q^n$, we have
    \begin{align*}
      \sum_{x \in S^\el} \chi_\al(x) =
      \begin{dcases*}
      |S|^\el, & if $S\le \Span(\al)^\perp$,\\
      0, & otherwise.
      \end{dcases*}
    \end{align*}
  \end{claim}

  \begin{proof}[Proof of~\cref{clm:fourier_subspace_sum}.]
    If $S\le\Span(\al)^\perp$, then all terms of the sum are $1$, so the result follows. On the
    other hand, if $S\not\le\Span(\al)^\perp$, then there exist $y\in S$ and $\beta\in\Span(\al)$
    such that $\chi_\beta(y)\neq 1$. Write $\beta = \sum_{j\in[\el]} c_j\cdot\al_j$ for $c_j\in\F_q$
    and let $z\in S^\ell$ be given by $z_j \coloneqq c_j\cdot y$ ($j\in[\el]$). Then we have
    \begin{align*}
      \sum_{x\in S^\el} \chi_\al(x)
      & =
      \sum_{x\in S^\el} \chi_\al(x + z)
      =
      \chi_\al(z)\cdot\sum_{x\in S^\el} \chi_\al(x)
      =
      \chi_\beta(y)\cdot\sum_{x\in S^\el} \chi_\al(x),
    \end{align*}
    and since $\chi_\beta(y)\neq 1$, we conclude that $\sum_{x\in S^\el} \chi_\al(x) = 0$.
  \end{proof}

  From~\cref{clm:fourier_subspace_sum} above, it follows that the Fourier constraint for $\al$ is
  equivalent to
  \begin{align*}
    \sum_{S \leq \Span(\al)^\perp}|S|^\el\pPr[S] \geq 0,
  \end{align*}
  concluding the proof.
\end{proof}

It will also be convenient to consider a weakening of this formulation that is more amenable to
analysis. Let $k_0 \coloneqq \log_q( A_q^\Lin(n,d))$ be the maximum dimension of a linear code of minimum
distance at least $d$. The program of~\cref{fig:pseudoprobabilities} below is obtained from that
of~\cref{fig:pseudoprobabilities_orig} by replacing the distance constraints with the following
``dimension constraints''.  {
\begin{empheq}[box=\fbox]{align*}
  \qquad & a_{(x_1, \dots, x_\el)} = 0
  & & \text{if }\dim(\Span(x_1, \dots, x_\el)) > k_0
  & & (\text{Dimension constraints})
\end{empheq}}

\begin{figure}[!htb]
{\small
\begin{empheq}[box=\fbox]{align*}
  & \text{Variables: } \pPr[S] && S \leq \F_q^n \\
  \max \quad
  & \sum_{S \le \F_q^n} \abs{S}^\el \pPr[S]
  \\
  \text{s.t.} \quad
  & \sum_{S \leq \F_q^n} \pPr[S] = 1
  & &
  & & (\text{Normalization})
  \\
  & \pPr[S] = 0
  & & \text{if } \dim(S) > k_0
  & & (\text{Dimension constraints})
  \\
  & \sum_{S \leq U} \abs{S}^\el\pPr[S] \geq 0
  & & \forall U \leq \F_q^n
  & & (\text{Fourier coefficients})
  \\
  & \sum_{S \geq U}\pPr[S] \geq 0
  & & \forall U \leq \F_q^n
  & & (\text{Nonnegativity}).
\end{empheq}}
 \caption{$\KLP^{\F_q}_\Lin(n,d,\el)$, weakened to dimension constraints, in terms of pseudoprobabilities for $\el \geq n$.}
 \label{fig:pseudoprobabilities}
\end{figure}

\begin{lemma}\label{lemma:kraw_pseudoprob}
  The program in~\cref{fig:pseudoprobabilities} is a relaxation of $\KLP^{\F_q}_\Lin(n,d,\el)$.
\end{lemma}

\begin{proof}
  Since $k_0\coloneqq\log_q( A_q^\Lin(n,d))$, any subspace of dimension larger than $k_0$ must have
  minimum distance less than $d$, so the distance constraints imply the dimension constraints. Thus,
  the result follows.
\end{proof}

From~\cref{lem:pseudoprob,lemma:kraw_pseudoprob}, to show exact completeness of
$\KLP^{\F_q}_\Lin(n,d,\el)$, it suffices to show that the weakened program
of~\cref{fig:pseudoprobabilities} has optimum value $A_q^\Lin(n,d)^\el$. The advantage of working
with the formulations that use the variables $\pPr[S]$ is that the Fourier constraints no longer
have sign alternations. However, the challenge is now to show that optimum solutions must force
$\pPr[S]$ to take nonnegative values.

\subsection{Exact Completeness Proof}\label{subsec:completeness_kt}

Before we start the proof, note that by level $n$ there is a variable for each possible basis of a
subspace of $\F_q^n$, which means that just writing down the distance constraints of the program
$\KLP^{\F_q}_\Lin(n,d,n)$ allows one to deduce the true value of $A_q^\Lin(n,d)$. However, the LP
hierarchy does not know how to use this kind of reasoning, hence our proof of completeness is more
involved. On the other hand, a feature of this subspace formulation of the hierarchy is that the
number of variables and constraints remains constant regardless of the level $\ell$ (as long as
$\ell\geq n$).

Note that we \emph{do not} show that the polytope is integral, meaning that feasible solutions are
integral (i.e., convex combinations of true solutions). In fact, we will see
in~\cref{prop:notintegral} that the polytope is not integral when $k_0\geq 2$.

We now restate and prove our main result.

\begin{theorem}\label{thm:completeness}
  For $\ell \ge n$, we have $A_q^\Lin(n,d) = \val(\KLP_\Lin^{\F_q}(n,d,\ell))^{1/\ell}$. More
  precisely, every $\GL_\el(\F_q)$-invariant optimum solution of $\KLP_\Lin^{\F_q}(n,d,\ell)$ is
  integral.
\end{theorem}

\begin{proof}[Proof of \cref{thm:completeness}]
  Since the program $\KLP_\Lin^{\F_q}(n,d,\el)$ is $\GL_\el(\F_q)$-invariant, the first assertion
  follows from the second assertion.
  
  An immediate consequence of~\cref{lem:mobius,lem:pseudoprob} is that to show integrality of
  $\GL_\el(\F_q)$-invariant optimum solutions of $\KLP_\Lin^{\F_q}(n,d,\el)$, it is sufficient to
  prove that every optimum solution $\pPr[S]$ ($S\le\F_q^n$) of the program
  in~\cref{fig:pseudoprobabilities_orig} is a probability distribution.

  Now we claim that it is sufficient to prove that every optimum solution of the program
  in~\cref{fig:pseudoprobabilities} is a probability distribution. Indeed, if this is the case, then
  the optimum value of both programs in~\cref{fig:pseudoprobabilities_orig,fig:pseudoprobabilities}
  must be $\lvert\F_q\rvert^{k_0\cdot\el}=A_q^\Lin(n,d)^\el$, since the definition of $k_0$ implies
  that there must be at least one true solution corresponding to a code $\cC$ of dimension $k_0$ and
  minimum distance at least $d$. In particular, every optimum solution of the former program must
  also be an optimum solution of the latter, hence a probability distribution.

  Let us then show that an optimum solution $\pPr$ of the program~\cref{fig:pseudoprobabilities} is
  a probability distribution. Since $\sum_{S\le\F_q^n}\pPr[S]=1$ already follows from the
  normalization constraint, we only have to show that $\pPr$ is nonnegative.

  If $S$ is a space in the support of $\pPr$ of minimum dimension, then $\pPr[S] \geq 0$ by the
  Fourier constraint on $S$. Thus, to show that $\pPr$ is nonnegative, it suffices to show that
  every such space of minimum dimension has dimension exactly $k_0$ (note that spaces of dimension
  larger than $k_0$ are not in the support of $\pPr$ due to the dimension constraints). To that end,
  let $S_{\min}$ be a subspace of minimum dimension in the support of $\pPr$, assume for the sake of
  contradiction that $\dim(S_{\min}) < k_0$ and let us show that there is a way to increase the
  objective value of $\pPr$. Indeed, we construct another solution $\widetilde{\Psymb}_+$ by
  transferring the probability mass from $S_{\min}$ and dividing it equally among the $S > S_{\min}$
  with $\dim(S) = \dim(S_{\min}) + 1$. Formally, letting $\mathcal{S} \coloneqq \{S \geq S_{\min} :
  \dim(S) = \dim(S_{\min}) + 1\}$ and $m \coloneqq \abs{\mathcal{S}}$ be the number of such spaces,
  we define:
  \begin{align*}
    \widetilde{\Psymb}_+[S]
    & \coloneqq
    \begin{dcases*}
      0, & if $S = S_{\min}$,\\
      \pPr[S] + \frac{\pPr[S_{\min}]}{m}, &
      if $S \geq S_{\min}$ and $\dim(S) = \dim(S_{\min}) + 1$,\\
      \pPr[S], & otherwise.
    \end{dcases*}
  \end{align*}

  Let us verify that $\pPr_+$ remains a feasible solution.
  \begin{itemize}
  \item $\pPr_+$ respects the normalization $\sum_{S \leq \F_q^n} \pPr_+[S] = 1$.
  \item The dimension constraints are not violated since $\dim(S) = \dim(S_{\min}) + 1 \leq k_0$ for
    the spaces $S \in \mathcal{S}$ in the second case above.
  \item In Fourier constraints with $U \not\geq S_{\min}$, nothing changes. In the ones with $U =
    S_{\min}$, the left-hand side is 0. Finally, when $U > S_{\min}$, $U$ contains at least one of
    the subspaces $S\in\mathcal{S}$ with increased mass. Therefore the change in the left-hand side
    is at least
    \begin{align*}
      \abs{S}^{\el}\cdot\frac{\pPr[S_{\min}]}{m}- \abs{S_{\min}}^\el\cdot\pPr[S_{\min}]\,.
    \end{align*}
    Since $m \leq \abs{\F_q}^{n}$ while $\frac{|S|^\el}{\abs{S_{\min}}^\el} = \abs{\F_q}^{\el} \geq
    \abs{\F_q}^n$, this is nonnegative.
  \item In the nonnegativity constraints, if $U$ is not below any space in $\mathcal{S}$, then nothing
    changes. If $U$ is below $S_{\min}$, then the sum in the nonnegativity constraint is unchanged
    since all $S \in \mathcal{S}$ appear in the sum.  Finally, if $U \in \mathcal{S}$, then the sum
    increased by $\pPr[S_{\min}]/m$.
  \end{itemize}

  Finally, note that objective value of the new solution ${\widetilde{\Psymb}_+}[S]$ is
  \begin{align*}
    \sum_{S \le \F_q^n} \abs{S}^\el {\widetilde{\Psymb}_+}[S]
    =
    \sum_{S \le \F_q^n} \abs{S}^\el \pPr[S] + \abs{S_{\min}}^\ell \(\abs{\F_q}^\ell - 1\) \pPr[S_{\min}],
  \end{align*}
  which is strictly larger than the previous objective value since
  $\pPr[S_{\min}] > 0$, a contradiction.

  Therefore, $\pPr$ must be supported only on spaces of dimension exactly $k_0$, it is nonnegative and
  integral, and the proof is complete.
\end{proof}


\section{Exact Completeness of the Partial Krawtchouk LP Hierarchy}\label{sec:partial_kraw_comple}

The hierarchy from~\cite{LL22} differs from the one in the previous section in two ways. Firstly, besides the Fourier
constraints, it includes the following partial Fourier
constraints:
\begin{empheq}[box=\fbox]{align*}
  & \sum_{x \in (\F_q^n)^\el} a_x  \theta_\alpha(x) \ge 0
  & & \forall \alpha \in (\F_q^n)^\el, \theta_\alpha \coloneqq \theta_{\alpha_1} \otimes \cdots \otimes \theta_{\alpha_\ell}, \theta_{\alpha_i} \in \{\chi_{\alpha_i},\one_{\alpha_i}\}
  & & (\text{Partial Fourier})
\end{empheq}
In the expression above, $\one_{\alpha_i} \colon \F_q^n \to
\F_q^n$ is the indicator function of $\alpha_i$. Secondly,
its objective function is slightly different, meant to be a
relaxation for the value $A_q(n,d)$ rather than
$A_q(n,d)^\ell$.

We denote by $\PKLP^{\F_q}_\Lin(n,d,\el)$ the level $\ell$
of the partial Krawtchouk hierarchy for $A_q(n,d)$
from~\cite{LL22}. An unsymmetrized version of this hierarchy
is presented in~\cref{program:partial_kt}. The exact
description of the hierarchy factors $\GL_\el(\F_q)$ and
$S_n$ symmetries (see also~\cref{lp:pklp}).

\begin{figure}[!htb]
{\small
\begin{empheq}[box=\fbox]{align*}
  & \text{Variables: } a_x && x \in (\F_q^n)^\el \\
  \max \quad
  & \sum_{x_1 \in \F_q^n} a_{(x_1,0,\ldots,0)}
  \\
  \text{s.t.} \quad
  & a_{0} = 1
  & &
  & & (\text{Normalization})
  \\
  & a_{(x_1, \dots, x_\el)} = 0
  & & \exists w \in \Span(x_1, \dots, x_\el).\; \abs{w} \in [d-1]
  & & (\text{Distance constraints})
  \\
  & \sum_{x \in (\F_q^n)^\el} a_x  \theta_\alpha(x) \ge 0
  & & \forall \alpha \in (\F_q^n)^\el, \theta_\alpha \coloneqq \theta_{\alpha_1} \otimes \cdots \otimes \theta_{\alpha_\ell}, \theta_{\alpha_i} \in \{\chi_{\alpha_i},\one_{\alpha_i}\}
  & & (\text{Partial Fourier})  
  \\
  & a_x = a_y 
  & & \forall x,y \in (\F_q^n)^\el, \Span(x) = \Span(y)
  & & (\text{$\GL_\ell(\F_q)$-symmetries})
  \\
  & a_x \geq 0
  & & \forall x \in (\F_q^n)^\el
  & & (\text{Nonnegativity}).
\end{empheq}}
 \caption{Unsymmetrized partial Krawtchouk hierarchy $\PKLP^{\F_q}_\Lin(n,d,\el)$.}
 \label{program:partial_kt}
\end{figure}

To show exact completeness of the partial Krawtchouk hierarchy, we will first give an alternative
description in terms of pseudoprobabilities in a similar way as done for the Krawtchouk hierarchy
in~\cref{subsec:pseudoprob_formulation}. It is enough to show exact completeness for the following
weakening given in~\cref{program:partial_kt}, where only full Fourier constraints are included. Note
that $\FKLP^{\F_q}_\Lin(n,d,\el)$ is the same as hierarchy of~\cref{lp:unsym_klp}
from~\cref{sec:kraw_comple} with a different objective function.

\begin{figure}[!htb]
{\small
\begin{empheq}[box=\fbox]{align*}
  & \text{Variables: } a_x && x \in (\F_q^n)^\el \\
  \max \quad
  & \sum_{x_1 \in \F_q^n} a_{(x_1,0,\ldots,0)}
  \\
  \text{s.t.} \quad
  & a_{0} = 1
  & &
  & & (\text{Normalization})
  \\
  & a_{(x_1, \dots, x_\el)} = 0
  & & \exists w \in \Span(x_1, \dots, x_\el).\; \abs{w} \in [d-1]
  & & (\text{Distance constraints})
  \\
  & \sum_{x \in (\F_q^n)^\el} a_x\chi_\alpha(x) \ge 0
  & & \forall \alpha \in (\F_q^n)^\el
  & & (\text{Full Fourier})
  \\
  & a_x = a_y 
  & & \forall x,y \in (\F_q^n)^\el, \Span(x) = \Span(y)
  & & (\text{$\GL_\ell(\F_q)$-symmetries})  
  \\
  & a_x \geq 0
  & & \forall x \in (\F_q^n)^\el
  & & (\text{Nonnegativity}).
\end{empheq}}
 \caption{Unsymmetrized partial Krawtchouk hierarchy
   $\FKLP^{\F_q}_\Lin(n,d,\el)$, weakened to only full Fourier constraints.}
 \label{program:full_kt}
\end{figure}

\subsection{A Pseudoprobability LP Formulation}\label{subsec:pseudoprod_partial}

Similarly to~\cref{subsec:pseudoprob_formulation}, we will show that the program
in~\cref{program:pseudoprobabilities_partial_kt} is a reformulation of the weakening of the partial
Krawtchouk hierarchy $\FKLP^{\F_q}_\Lin(n,d,\el)$.

\begin{figure}[!htb]
{\small
\begin{empheq}[box=\fbox]{align*}
  & \text{Variables: } \pPr[S] && S \leq \F_q^n \\
  \max \quad
  & \sum_{S \le \F_q^n} \abs{S} \pPr[S]
  \\
  \text{s.t.} \quad
  & \sum_{S \leq \F_q^n} \pPr[S] = 1
  & &
  & & (\text{Normalization})
  \\
  & \pPr[S] = 0
  & & \text{if } \dim(S) > k_0
  & & (\text{Dimension constraints})
  \\
  & \sum_{S \leq U} \abs{S}^\el\pPr[S] \geq 0
  & & \forall U \leq \F_q^n
  & & (\text{Fourier coefficients})
  \\
  & \sum_{S \geq U}\pPr[S] \geq 0
  & & \forall U \leq \F_q^n
  & & (\text{Nonnegativity}).
\end{empheq}}
 \caption{$\FKLP^{\F_q}_\Lin(n,d,\el)$, weakened to dimension constraints and
   with only the full Fourier constraints, in terms of pseudoprobabilities for
   $\el \geq n$.}
 \label{program:pseudoprobabilities_partial_kt}
\end{figure}

\begin{lemma}\label{lemma:pseudoprob_partial_kraw}
  If $a$ is a $\GL_\el(\F_q)$-invariant solution of $\FKLP^{\F_q}_\Lin(n,d,\el)$
  from~\cref{program:full_kt} and $\pPr[\Span(x)\subseteq\csos]\coloneqq a_x$ for
  every $x\in(\F_q^n)^\el$, then $\pPr[S = \csos]$ given
  by~\cref{lem:mobius}\ref{lem:mobius:inverse} is a solution of the program
  in~\cref{program:pseudoprobabilities_partial_kt} with the same value.

  Conversely, if $\pPr[S=\csos]$ is a solution of the program
  in~\cref{program:pseudoprobabilities_partial_kt}, then setting
  $a_x\coloneqq\pPr[\Span(x)\subseteq\csos]$
  via~\cref{lem:mobius}\ref{lem:mobius:direct} gives a solution of
  $\FKLP^{\F_q}_\Lin(n,d,\el)$ with the same value.
\end{lemma}

\begin{proof}
  Since the program of~\cref{program:partial_kt} is the same as the program of~\cref{lp:unsym_klp}
  except for the objective function, the proof is the same as that of~\cref{lem:pseudoprob} only
  differing in the objective function analysis. But note that the rewritten objective function is
  \begin{align*}
    \sum_{x=(x_1,0,\ldots,0) : x_1 \in \F_q^n} \pPr[\Span(x) \subseteq \csos]
    & =
    \sum_{x=(x_1,0,\ldots,0) : x_1 \in \F_q^n} \sum_{\Span(x)\le T\le\F_q^n}\pPr[T]
    =
    \sum_{S \le \F_q^n} \abs{S} \pPr[S],
  \end{align*}
  concluding the proof.
\end{proof}

For the exact completeness, it will be sufficient to consider the above weakened pseudoprobability
formulation of~\cref{lemma:pseudoprob_partial_kraw}. However, to cover some integrality properties
of~\cref{sec:integrality}, it will also be useful to give a pseudoprobability formulation of
$\PKLP^{\F_q}_\Lin(n,d,\el)$ that includes all partial Fourier constraints.

\begin{figure}[!htb]
{\small
\begin{empheq}[box=\fbox]{align*}
  & \text{Variables: } \pPr[S] && S \leq \F_q^n \\
  \max \quad
  & \sum_{S \le \F_q^n} \abs{S} \pPr[S]
  \\
  \text{s.t.} \quad
  & \sum_{S \leq \F_q^n} \pPr[S] = 1
  & &
  & & (\text{Normalization})
  \\
  & \pPr[S] = 0
  & & \exists w \in S.\; \abs{w} \in [d-1] 
  & & (\text{Distance Constraints})
  \\
  & \sum_{T \leq S \leq U} \abs{S}^r \pPr[S] \geq 0
  & & \forall T \le U \leq \F_q^n : \substack{n-\dim(U) \le r \le \ell,\\ \dim(T) \le \ell -r}
  & & (\text{Partial Fourier coefficients})
  \\
  & \sum_{S \geq U}\pPr[S] \geq 0
  & & \forall U \leq \F_q^n
  & & (\text{Nonnegativity}).
\end{empheq}}
 \caption{$\KLP^{\F_q}_\Lin(n,d,\el)$ in terms of pseudoprobabilities for $\el \geq n$.}
 \label{fig:pseudoprobabilities_partial_orig}
\end{figure}

\begin{lemma}\label{lem:pseudoprob_full}
  If $a$ is a $\GL_\el(\F_q)$-invariant solution of the program $\PKLP^{\F_q}_\Lin(n,d,\el)$
  in~\cref{program:partial_kt} and $\pPr[\Span(x)\subseteq\csos]\coloneqq a_x$ for every
  $x\in(\F_q^n)^\el$, then $\pPr[S = \csos]$ given by~\cref{lem:mobius}\ref{lem:mobius:inverse} is a
  solution of the program in~\cref{fig:pseudoprobabilities_partial_orig} with the same value.

  Conversely, if $\pPr[S=\csos]$ is a solution of the program
  in~\cref{fig:pseudoprobabilities_partial_orig}, then setting
  $a_x\coloneqq\pPr[\Span(x)\subseteq\csos]$ via~\cref{lem:mobius}\ref{lem:mobius:direct} gives a
  solution of $\PKLP^{\F_q}_\Lin(n,d,\el)$ with the same value.
\end{lemma}

\begin{proof}
  The program in~\cref{fig:pseudoprobabilities_partial_orig} is the same as the program
  in~\cref{program:pseudoprobabilities_partial_kt} with additional partial Fourier constraints. We
  can then follow the proof of~\cref{lem:pseudoprob} except for these additional constraints which
  we now analyze.

  For $\mathcal{I} \subseteq [\ell]$ and $\al\in(\FF_q^n)^\el$, let
  $\theta_\al\coloneqq\theta_{\alpha_1} \otimes \cdots \otimes \theta_{\alpha_\ell}$, where
  $\theta_{\alpha_i}\coloneqq\chi_{\alpha_i}$ for $i \in \mathcal{I}$ and
  $\theta_{\alpha_i}\coloneqq\one_{\alpha_i}$ for $i \in [\ell] \setminus \mathcal{I}$. The
  left-hand side of the Fourier constraint for $\theta_\al$ is
  \begin{align*}
    \sum_{x \in (\F_q^n)^\el} \pPr[\Span(x) \subseteq \csos] \theta_\al(x)
    & =
    \sum_{x \in (\F_q^n)^\el} \theta_\al(x) \sum_{T \geq \Span(x)}\pPr[T]
    =
    \sum_{S \leq \F_q^n} \pPr[S] \sum_{\substack{x \in S^\el}} \theta_\al(x).
  \end{align*}

  We now prove the following claim.
  \begin{claim}\label{clm:partial_fourier_sum}
    For $\mathcal{I} \subseteq [\ell]$, $\al \in (\F_q^n)^\el$ and $S \leq \F_q^n$, we have
    \begin{align*}
      \sum_{x \in S^\el} \theta_\al(x)
      & =
      \begin{dcases*}
        |S|^{\abs{\mathcal{I}}} \prod_{j\in [\ell]\setminus \mathcal{I}} \one_{\al_j \in S},
        & if $S \leq \Span(\al_i : i \in \mathcal{I})^\perp$,
        \\
        0, & otherwise.
      \end{dcases*}
    \end{align*}
  \end{claim}

  \begin{proof}[Proof of~\cref{clm:partial_fourier_sum}.]
    Clearly, if there exists $j\in[\el]\setminus\mathcal{I}$ such that $\al_j\notin S$, then the sum
    above is zero. Suppose then that for every $j\in[\el]\setminus\mathcal{I}$, we have $\al_i\in
    S$. Then the sum becomes
    \begin{align*}
      \sum_{x \in S^{\mathcal{I}}} \prod_{j\in \mathcal{I}} \chi_{\al_j}(x_j),
    \end{align*}
    and the result follows by~\cref{clm:fourier_subspace_sum}.
  \end{proof}

  Let $U \coloneqq \Span(\alpha_i : i \in \mathcal{I})^\perp$ and let $T \coloneqq
  \Span(\alpha_i : i \in [\ell] \setminus \mathcal{I})$. By~\cref{clm:partial_fourier_sum}, the
  Fourier constraint corresponding to $\theta_\al$ is equivalent to
  \begin{align*}
    \sum_{T \le S \leq U}|S|^{\abs{\mathcal{I}}} \pPr[S] \geq 0.
  \end{align*}

  Since the program in~\cref{fig:pseudoprobabilities_partial_orig} has the partial Fourier constraints
  \begin{align*}
    \sum_{T \leq S \leq U} \abs{S}^r \pPr[S] \geq 0
    \qquad \forall T \le U \leq \F_q^n \quad :
    \qquad \substack{n-\dim(U) \le r \le \ell,\\ \dim(T) \le \ell -r},
  \end{align*}
  it remains to show every $U$ and $T$ above can be obtained as $U = \Span(\alpha_i : i \in
  \mathcal{I})^\perp$ and $T = \Span(\alpha_i : i \in [\ell] \setminus \mathcal{I})$.

  Indeed, for every such $U \leq \F_q^n$ with $u\coloneqq\dim(U) \ge n-\ell$, we can use $r \in
  \{n-u,\ldots,\ell\}$ entries in the vector $\alpha$ to specify a spanning set for $U^\perp$. These
  entries will correspond to some $\mathcal{I} \subseteq [\ell]$ of size $r$. We then use the
  remaining $\ell-r$ entries of $\alpha$ to specify a spanning set for the space $T \le U$ of
  dimension at most $\ell - r$, concluding the proof.
\end{proof}

\subsection{Exact Completeness Proof}\label{sec:completeness_partial}

We now prove the exact completeness of the partial Krawtchouk hierarchy of~\cite{LL22}.

\begin{theorem}
  For $\ell \ge n$, we have $A_q^\Lin(n,d) = \val(\PKLP_\Lin^{\F_q}(n,d,\ell))$. More precisely,
  every $\GL_\el(\F_q)$-invariant optimum solution of $\PKLP_\Lin^{\F_q}(n,d,\ell)$ is integral.
\end{theorem}

\begin{proof}
  By~\cref{lemma:pseudoprob_partial_kraw} and similarly to~\cref{thm:completeness}, it is enough to
  show that every optimum solution $\pPr$ of the pseudoprobability program
  in~\cref{program:pseudoprobabilities_partial_kt} is nonnegative.

  Note that the feasible region of this pseudoprobability program is the same as the one of the
  pseudoprobability program of~\cref{fig:pseudoprobabilities}, so we can follow the same
  completeness proof of~\cref{thm:completeness}, except for the objective function
  analysis. Inspecting that proof, we see that it only requires the property that the objective
  value increases if mass is moved to larger dimensional spaces. This property is also satisfied by
  the new objective function $\sum_{S \le \F_q^n} \abs{S} \pPr[S]$ so we are done.
\end{proof}


\section{On Integrality Related Properties}
\label{sec:integrality}

In this section, we discuss some properties related to the integrality of the Krawtchouk
hierarchies. Recall that by the results of~\cref{sec:kraw_comple,sec:partial_kraw_comple},
integrality of $\GL_\el(\F_q)$-invariant solutions is equivalent to nonnegativity of solutions in the
pseudoprobability formulations; as such, we will slightly abuse notation and say that the polytope
of the pseudoprobability formulation is integral when all its feasible solutions are nonnegative.

We start by showing that the polytope of the pseudoprobability formulation of the program
$\KLP_\Lin^{\F_q}(n,d,\ell)$ is not integral no matter how large is the level of the hierarchy.

\begin{proposition}\label{prop:notintegral}
  The polytope defined by the constraints of the pseudoprobability formulation
  from \cref{fig:pseudoprobabilities_orig} is not integral for any $k_0 \ge 2$.
\end{proposition}

\begin{proof}
  We construct a feasible solution to the program in~\cref{fig:pseudoprobabilities}
  having a negative pseudoprobability.
  Let $T \le \F_q^n$ be any subspace of dimension $k_0$ of minimum distance at least $d$
  and let $T' \le T$ be an arbitrary one dimensional space. Since $k_0 \ge 2$, we have $T' \ne T$.
  Let $\epsilon \in (0,1)$. Now for each $S \le \F_q^n$, we set
  \begin{align*}
    \pPr[S]
    & \coloneqq
    \begin{dcases*}
      1 - \left(\epsilon- \frac{\epsilon}{\abs{S}^\ell}\right), &  if $S = T$,\\
      -\frac{\epsilon}{\abs{S}^\ell}, & if $S = T'$,\\
      \epsilon, & if $S = \{0\}$,\\                
      0, & otherwise.
    \end{dcases*}
  \end{align*}
  We claim that the above is a feasible solution. The proof is a
  simple verification. The values $\pPr[S]$ clearly sum to $1$
  satisfying the normalization constraint. Since $T$ has minimum distance at least $d$, so does
  $T'$, hence the distance constraints are satisfied. The Fourier constraint of $T'$
  is
  \begin{align*}
    \abs{T'}^\ell \pPr[T'] + \pPr[\{0\}] = 0.
  \end{align*}
  Since all values except from $\pPr[T']$ are nonnegative and the Fourier constraint
  of $T'$ is satisfied, we have that all Fourier constraints hold.
  The nonnegative constraint for $T'$ is
  \begin{align*}
    \sum_{S \geq T'}\pPr[S]
    & =
    \pPr[T] + \pPr[T']
    =
    1 - \left(\epsilon- \frac{\epsilon}{\abs{S}^\ell}\right) - \frac{\epsilon}{\abs{S}^\ell}
    =
    1-\epsilon
    \ge
    0,
  \end{align*}
  where the last inequality follows from our choice of $\epsilon$. All
  other nonnegative constraints are easily seen to hold and we
  conclude the proof.
\end{proof}

Despite the polytope not being integral no matter how large is the
level, the following approximate integrality property holds: any given
non-integral solution becomes infeasible at a sufficiently large
level. 

\begin{proposition}\label{prop:intersectionintegral}
  Let $\{\pPr[S]\}_{S \le \F_q^n}$ be a feasible solution to level $\ell$ of program
  \cref{fig:pseudoprobabilities_orig}. If one of the variables is negative, then there exist $\ell'
  \ge \ell$ large enough such that this solution is infeasible for level $\ell'$.
\end{proposition}

\begin{proof}
  Let $U \le \F_q^n$ be any space such that $\pPr[U] < 0$ and its dimension is maximum with this
  property. Note that $U$ is well-defined by assumption. We claim that the Fourier constraint
  \begin{align*}
    \sum_{S \leq U} \abs{S}^{\el'} \pPr[S] \ge 0
  \end{align*}
  becomes violated for a sufficiently large $\ell' \ge \ell$.
  By dividing this Fourier constraint by $\abs{S}^{\el'}$, only the coefficient of $\pPr[U]$
  remains $1$ while all other coefficients shrink as $\ell'$ grows since $U$ is the space
  of largest dimension appearing in the sum.
\end{proof}

Let us now show that the additional partial Fourier constraints ensure that the polytope of the
pseudoprobability formulation of the hierarchy $\PKLP^{\F_q}_\Lin(n,d,\el)$ is actually integral for
$\ell \ge n$. Note that this provides an alternative proof of exact completeness.

\begin{proposition}\label{prop:integ_ll_hier}
  The polytope defined by the constraints of the pseudoprobability formulation
  from \cref{fig:pseudoprobabilities_partial_orig} is integral for $\ell \ge n$.
\end{proposition}

\begin{proof}
  We will show that $\pPr[T] \ge 0$ for every $T \leq
  \F_q^n$. Combined with the normalization constraint $\sum_{T \le
    \F_q^n} \pPr[T] = 1$, we will have a true probability distribution
  over valid codes and thus the polytope will be integral. Recall the
  Fourier constraints
  from~\cref{fig:pseudoprobabilities_partial_orig},
  \begin{align*}
  \sum_{T \leq S \leq U} \abs{S}^r \pPr[S] \geq 0
  \qquad \forall T \le U \leq \F_q^n
  \quad :
  \qquad \substack{n-\dim(U) \le r \le \ell,\\ \dim(T) \le \ell -r}.
  \end{align*}  
  Since $\ell \ge n$, by choosing $r \coloneqq n -\dim(U)$, we can take $T\coloneqq U$. In this
  case, the sum above reduces to only the term $\pPr[U]$ with the coefficient $\abs{U}^r > 0$. This
  readily implies $\pPr[U] \ge 0$.
\end{proof}


\section{Conclusion}

In this paper, we proved exact completeness by level $n$ of the LP
hierarchies $\KLP$ and $\PKLP$ of~\cite{CJJ22} and~\cite{LL22},
respectively. Our techniques involved passing to a formulation of
these hierarchies in terms of pseudoprobabilities (after appropriate
symmetrization under the natural $\GL_\el(\F_q)$ action) and showing
that optimum solutions are integral (i.e., are convex combinations of
true solutions, corresponding to linear codes). We also observed two
structural properties about the feasible polytopes of these
hierarchies: while for $\KLP$ no level guarantees integrality of the
polytope, for $\PKLP$, the polytope is integral by level $n$.

\medskip

As mentioned before, the completeness results of these hierarchies
should be seen as theoretical results that can serve as basis for a
theoretical analysis of the asymptotic behavior of
$A^\Lin_q(n,d)$. However, neither of the hierarchies should be computationally run as high as level $\el = n$,
since even writing the constraints at this level 
involves checking which $\ell$-dimensional subspaces
satisfy the distance constraints. If $\ell >
k_0\coloneqq\log_q A^\Lin_q(n,d)$, then we would be able to deduce the
value of $k_0$ by simply noting that no subspace of dimension $k_0+1$
satisfies the distance constraints. This simple observation makes
plausible that the hierarchies could be complete by an earlier level,
say $O(k_0)$.

\section*{Acknowledgement}

We thank Avi Wigderson for stimulating discussions during
the initial phase of this project.

\bibliographystyle{alphaurl}
\bibliography{macros,references}

\end{document}
